\documentclass[12pt]{amsart}
\usepackage[margin=1in]{geometry}
\usepackage{mathptmx}
\usepackage{mathtools}
\usepackage{amsthm}
\usepackage{amssymb}
\usepackage[alphabetic]{amsrefs}
\usepackage{graphicx}
\usepackage{tikz}
	\usetikzlibrary{decorations.pathreplacing}
	\usetikzlibrary{decorations.pathmorphing}
	\usetikzlibrary{patterns}
\usepackage{caption}

\newcommand{\cd}{\cdot}

\newcommand{\C}{\mathbb{C}}

\newcommand{\R}{\mathbb{R}}
\newcommand{\Z}{\mathbb{Z}}

\newcommand{\lbar}[1]{\overline{#1}}
\newcommand{\id}{\mathrm{id}}

\newtheorem{thm}{Theorem}

\newtheorem{lemma}{Lemma}

\theoremstyle{definition}

\newtheorem*{defin*}{Definition}
\theoremstyle{remark}

\begin{document}

\title{Critical Metrics and Covering Number}

\author{Michael H. Freedman}
\address{\hskip-\parindent
	Michael H. Freedman \\
	Microsoft Research, Station Q, and Department of Mathematics \\
	University of California, Santa Barbara \\
	Santa Barbara, CA 93106
}

\begin{abstract}
	In both quantum computing and black hole physics, it is natural to regard some deformations, infinitesimal unitaries, as \emph{easy} and others as \emph{hard}. This has lead to a renewed examination of right-invariant metrics on $\operatorname{SU}(2^N)$. It has been hypothesized that there is a critical such metric---in the sense of phase transitions---and a conjectural form suggested. In this note we explore a restriction that the ring structure on cohomology places on the global geometry of a critical metric.
\end{abstract}


\maketitle

A recent paper \cite{BFLS21} examined the coarse geometry of the special unitary group $\operatorname{SU}(n)$, $n = 2^N$, with certain right invariant metrics related both to quantum computation and black hole dynamics. These metrics are called \emph{complexity geometries} as they attempt to quantify the difficulty in synthesizing operators. The chief purpose was to assemble evidence for the existence of a preferred \emph{critical} metric among these. In the critical metric, shortest paths (except at the smallest scales) become highly degenerate. This note employs a cohomological method to understand a geometric consequence of criticality. Two extreme measures of a space's geometry are total volume and diameter. For right invariant metrics on a lie group, understanding the scaling of volume is a triviality as it is proportional to $\sqrt{g}$ where $g$ is the determinant of the infinitesimal metric on the lie algebra, in this case $\operatorname{su}(n-1)$. Volume is, up to dimension-dependent constants, how many grains of sand it takes to fill, or $\epsilon$-balls to cover the space. Diameter is much more difficult to estimate. Taughtologically it is the smallest $d$ such that a single ball of diameter $d$ suffices to cover the space. Clearly these are merely limiting cases of a function, the \emph{covering function} with more descriptive power: Given a compact metric space $X$ and a positive $d$, define $C_X(d) \in \Z^+$ to be the minimal number of balls (in the metric) of diameter $\leq d$ required to cover $X$. Just as volume tells us how large a phase space is, and diameter tells us about the worst case difficulties of travel, $C_X(d)$ tells us how many distinct \emph{locales} of size $d$ exist in the space.

While total volume is easy, upper bounds on diameter are difficult but of great interest. Pioneering work on this problem by Nielsen and collaborators \cites{nielsen05,NDGH06a,NDGH06b,DN06} has, as one of its high points, a result \cite{NDGH06b}, that in one right invariant \emph{cliff} geometry, which exponentially punishes motion in any direction touched more than two qubits, that $\operatorname{diam}(SU(2^N))$ indeed scales exponentially in $N$. These methods are enhanced and extended in \cites{brown21,brown22}. Brown proves exponential scaling in $N$ of diameter for metrics on $\operatorname{SU}(n)$ with a variety of penalty schedules, including the schedule of Line \ref{eq:b_const} below, for all $b > 1$.

We examine the cohomological consequence of the criticality ansatz, which in particular implies exponential diameter, and find that a classical technique, Lusternik-Schnirelmann category, can be adapted to produce non-trivial information on the covering number function: $C_{\operatorname{SU}(n)_{\text{crit.}}}(d)$. This enriches our understanding of critical geometries, provides a potential technique for concluding a given geometry is \emph{not} critical, and opens the door to a more topological discussion of complexity lower bounds. With regard to the latter, it is worth recalling \cite{bss89} where counting Morse critical points in a computational space was used to lower-bound the number of algorithmic branch points. There, as in this note, even small factors are hard won, far from tight, but conceptually important due to their topological origins.

Let us recall the Brown-Susskind exponential penalty metrics \cite{bs18}. They may be described as follows. The Hermitian $2 \times 2$ matrices form a 4D real vector space spanned by the Pauli operators
\[
	\id = \begin{vmatrix}
		1 & 0 \\ 0 & 1
	\end{vmatrix},\ X = \begin{vmatrix}
		0 & 1 \\ 1 & 0
	\end{vmatrix},\ Y = \begin{vmatrix}
		0 & -i \\ i & 0
	\end{vmatrix}, \text{ and } Z = \begin{vmatrix}
		1 & 0 \\ 0 & -1
	\end{vmatrix}
\]

Similarly, it may be proved by induction on $N$ that $\operatorname{Herm}(2^N)$, the Hermitian operators on $N$ qubits $(\C^2)^{\otimes N}$ is the linear span of $4^N$ Pauli words, each of which is an ordered tensor product of $N$ factors chosen from $\{I, X, Y, Z\}$. So, for example $1 \otimes 1 \otimes Z \otimes X \otimes 1 \otimes Y$ is a Pauli word of length 6 and weight 3, weight $w$ being the number of non-identity entries. $\operatorname{Herm}_0(2^N)$ denotes the traceless Hermitians; these have dimension $4^N - 1$, the ``all 1's'' basis element being omitted. Multiplying by $i = \sqrt{-1}$, $i \operatorname{Herm}_0(2^N)$ is the traceless skew-Hermitians, that is the Lie algebra $\operatorname{su}(2^N)$. A right (or left) invariant metric on the simple Lie group $\operatorname{SU}(2^N)$ is specified by giving a metric, i.e.\ a nonsingular inner product, on $\operatorname{su}(2^N)$. Fixing any base $b > 1$, the BS exponential penalty metrics have the form:
\begin{equation}\label{eq:b_const}
	g_{IJ}^b \coloneqq \delta_{IJ} b^{2w(I)}
\end{equation}

That is, these metrics are diagonal in the Pauli word basis and exponentially penalize many-body interaction. The weight of a Pauli word being how many qubits it couples together. As referenced above,
\begin{equation}
	\operatorname{diam}\big(\operatorname{SU}(2^N), g_{IJ}^b\big) \text{ is exponential in } N\text{, for } b > 1
\end{equation}

Our ansatz is that for some base $b$, $g_{IJ}^b$ is critical. Combinatorial gate counting arguments \cite{BFLS21} suggest $b = 4$ may be the critical value, but the exact value will not be important. Let $d_{\text{crit}}$ denote the diameter of the critical metric, and $\beta$ be a geodesic arc between a pair of maximally distant points.

For any path metric space $X$, so in particular any Riemannian metric on a connected manifold, there is a trivial lower bound to the covering number $C(d):$
\begin{equation}
	C_X(d) \geq \frac{\operatorname{diam}(X)}{d} \coloneqq C^t(d), \text{ the \emph{trivial} lower bound}
\end{equation}
coming from solving the 1D problem of covering a shortest arc joining the most distant pair of points $(x,y)$ by sets of diameter $d$. We use $t(d)$ to denote the trivial estimate.

While the geodesic geometry of $g_{IJ}^b$ on $\operatorname{SU}(2^N)$ is enormously complicated, encoding, in a sense, all answers to problems of optimal quantum control, it is conjectured in \cite{BFLS21} that the base $b$ may be tuned to a critical value in which the IR behavior becomes simple: that for $b = b_0$ large distances only experience tiny (sub leading) corrections as the metric is \emph{increased} on some (or all) Pauli words of weight $> 2$, but may experience dramatic reductions of distance if the metric in some of the high weight directions is decreased. The idea is that at this critical exponent $b_0$, massive geodesic redundancy occurs; one can get just about anywhere (that is not too close!) just as efficiently by zig-zagging about in 2-local directions as by following more direct but locally more expensive multibody directions. Adding costs to the multibody directions changes large  distances only slightly because there exist 2-body alternative strategies. But lowering costs to any multibody direction will break the tie and lead to dramatic changes in the IR. While not proven, evidence from low-dimensional models, gate models, and other consistency checks are given.

In this note the existence of a critical $g_{IJ}^{b_0}$ is taken as an ansatz from which to explore geometric implications, to wit the covering number function, $C(d)$.

Let $F(k,N)$ be the cumulative binomial distribution function:
\begin{equation}
	F(k,N) = \sum_{i=1}^k \binom{N}{i}
\end{equation}

This function counts the dimensions of a family of tori $\subset \operatorname{SU}(n)$, exhausting a maximal torus by successively including maximal commuting operators on $k$ and fewer bodies. A concrete way of doing this is to build a sequence of Cartan subalgebras of $\operatorname{su}(n)$:
\begin{equation}\label{eq:strings}
\begin{split}
	& \langle X 1 1 \cdots 1, 1 X 1 \cdots 1, 1 1 X \cdots 1, 1 1 1 \cdots X \rangle,\ \dim = N \\
	& \langle X 1 1 \cdots 1, 1 X 1 \cdots 1, 1 1 X \cdots 1, 1 1 1 \cdots X; X X 1 \cdots 1, \dots, 1 1 1 \cdots X X \rangle,\ \dim = F(2,N) \\
	& \cdots,\ \dim = F(k,n)
\end{split}
\end{equation}
and the exponentiate these to nested tori $T_k$ in $\operatorname{SU}(n)$.

Abstractly, it is easy to compute the diameters using the restriction of the metric $g_{IJ}^b$ of these metric tori $T_k$ (as spaces in their own right, ignoring possible shortcuts through $\operatorname{SU}(n)$) the result is:
\begin{equation}\label{eq:tori_diam}
	\operatorname{diam}(T_k) = \pi \sqrt{\sum_{i=1}^k \binom{N}{i} b^{2i}} =: e_k, \text{ the diameters of } \leq k \text{-body tori}
\end{equation}

A consequence of the criticality ansatz is that there will be no (substantial) shortcuts and, in fact, $e_n$ lower bounds $\operatorname{diam}_c(\operatorname{SU}(n)) =: d_c$, the subscript $c$ denoting our critical $g_{IJ}^{b_0}$.

The equations of motion on a Lie group with right-invariant metric, the Arnold-Keshin equation (see Line 46 \cite{brown22})
\begin{equation}
	\langle \dot{H},K\rangle = i\langle H,[H,K]\rangle
\end{equation}
describes how a geodesic headed in direction $H$ appears to turn towards an arbitrary direction $K$ in terms of the Lie bracket $[,]$ and metric $\langle,\rangle$. $H,K \in \operatorname{su}(n)$ and the differential of right multiplication is used to identify tangent vectors to $\operatorname{SU}(n)$ with the Lie algebra $\operatorname{su}(n)$.

From (\ref{eq:tori_diam}) it is apparent that these tori are all totaly geodesic, for if $H$ is a superposition of strings as in (\ref{eq:strings}), the bracket with our arbitrary $K$ will not contain any similar string in its explansion and thus be orthogonal to $H$ w.r.t.\ any diagonal metric such as (\ref{eq:b_const}). (For example, if $H = 1 \otimes X \otimes X$ and $K = X \otimes Y \otimes X$, $[H,K] = X \otimes Z \otimes 1 - (-X \otimes Z \otimes 1) = 2 X \otimes Z \otimes 1$.) Thus, the tori $T_k$ are all totally geodesic. This means that \emph{locally} there are no \emph{shortcuts} through $\operatorname{SU}(n)$ between pairs of points on $T_k$. The criticality ansatz promotes this to a global statement. Up to sub-leading conditions, which we ignore, synthetic paths \emph{through} $\operatorname{SU}(n)$ between points of $T_k$ should be at best degenerate with geodesic arcs \emph{on} $T_k$.

We next consider the covering function $C_{T_k}$ at $d = b^k$, a value slightly smaller than new (cardinal direction) circle subgroups of $T_k$ complementary to $T_{k-1}$ (which have diameter $\pi b^k$).

\begin{lemma}\label{lm:covering}
	$C_{T_k}(b^k) \geq \binom{N}{k}$.
\end{lemma}

\begin{proof}
	$H^\ast(T_k;\Z)$ is an exterior algebra on $F(k,N)$ 1-dimensional generators. This follows from the K\"{u}nneth formula and induction as a torus is a product of circles. Similarly, $H^\ast(T_{k-1};\Z)$ is an exterior algebra on $F(k-1,N)$ 1-dimensional generators and since restricting induces an injection. The relative cohomology group $H^\ast(T_k, T_{k-1};\Z)$ is an exterior algebra on $\binom{N}{k}$ generators $\{\alpha_1, \dots, \alpha_{\binom{N}{k}}\}$ which may be identified as pairing $\delta_{ij}$ with the $\binom{N}{k}$ ``long circles'': $e^{2\pi i t}($weight $k$ word of $X$'s$) \subset T_k$.

	Now, to obtain a contradiction, suppose $\{\mathcal{O}_1, \dots, \mathcal{O}_{\binom{N}{k}}\}$ is a covering of $T_k$ by $\binom{N}{k}$ sets of diameter $\leq b^k$, a value chosen too small to contain any of the long circles subgroups. The diagram below is with integer coefficients and all arrows are induced by restriction classes (the reader may take real coefficients and, via de Rham theory, think of restricing differential forms if she prefers).
	\[
		\begin{tikzpicture}
			\node at (0,0) {$H^1(T_k, T_{k-1} \cup \mathcal{O}_1) \otimes H^1(T_k, T_{k-1} \cup \mathcal{O}_2) \otimes \cdots \otimes H^1(T_k, T_{k-1} \cup \mathcal{O}_{\binom{N}{k}}) \xrightarrow{\cup} H^{\binom{N}{k}}(T_k,\bigcup_{i=1}^{\binom{N}{k}} \mathcal{O}_i) \cong 0$};
			
			\node at (-0.6,-1) {$\otimes \cdots \otimes$};
			\node at (1.8,-1) {$H^1(T_k,T_{k-1})$};
			\node at (3.8,-1) {$\xrightarrow{\cup}$};
			\node at (5.6,-1) {$H^{\binom{N}{k}}(T_k, T_{k-1})$};
			\node at (7.45,-1) {$\not\cong 0$};
			\node at (-2.6,-1) {$H^1(T_k,T_{k-1})$};
			\node at (-6.1,-1) {$H^1(T_k,T_{k-1})$};
			\node at (-4.5,-1) {$\otimes$};
			
			\node at (-0.6,-2) {$\otimes \cdots \otimes$};
			\node at (1.8,-2) {$H^1(\mathcal{O}_{\binom{N}{k}})$};
			\node at (-2.6,-2) {$H^1(\mathcal{O}_2)$};
			\node at (-4.5,-2) {$\otimes$};
			\node at (-6.1,-2) {$H^1(\mathcal{O}_1)$};
			
			\draw[->] (1.8,-0.3) -- (1.8,-0.7);
			\draw[->] (1.8,-1.3) -- (1.8,-1.7);
			\draw[->] (5.6,-0.3) -- (5.6,-0.7);
			\draw[->] (-2.6,-0.3) -- (-2.6,-0.7);
			\draw[->] (-2.6,-1.3) -- (-2.6,-1.7);
			\draw[->] (-6.1,-0.3) -- (-6.1,-0.7);
			\draw[->] (-6.1,-1.3) -- (-6.1,-1.7);
			\node[rotate=20] at (-7.6,-0.4) {\footnotesize{$\lbar{\alpha}_1 \in$}};
			\node[rotate=20] at (-4.15,-0.4) {\footnotesize{$\lbar{\alpha}_2 \in$}};
			\node[rotate=20] at (0.1,-0.45) {\footnotesize{$\lbar{\alpha}_{\binom{N}{k}} \in$}};
			\node[rotate=20] at (-7,-1.3) {\footnotesize{$\alpha_1 \in$}};
			\node[rotate=20] at (-3.45,-1.3) {\footnotesize{$\alpha_2 \in$}};
			\node[rotate=20] at (0.8,-1.45) {\footnotesize{$\alpha_{\binom{N}{k}} \in$}};
		\end{tikzpicture}
	\]

	Because the $\mathcal{O}_i$ are too small to contain the dual circle subgroups, the $\alpha_i$ restrict to zero and thus pull back to $\lbar{\alpha}_i$. However, the $\lbar{\alpha}_i$ must cup to zero if $\{\mathcal{O}_i\}$ is a cover of $T_k$ as the product lands in the cohomology of the trivial pair $(T_k, \bigcup_{i=1}^{\binom{N}{k}} \mathcal{O}_i) = (T_k, T_k)$. This contradicts naturality of the vertical map to $H^{\binom{N}{k}}(T_k, T_{k-1})$ where the product is the top class in its exterior algebra.
\end{proof}

In the context of the ansatz, Lemma \ref{lm:covering} can be promoted to a lower bound on $C_{\operatorname{SU}(n)}(b^k)$, using the 1D reasoning that lead to $C^t(d)$, the trivial lower bound. By Line \ref{eq:tori_diam} and the discussion below it, $\operatorname{SU}(n)$ contains a geodesic arc $\alpha$ of length $d_c$. Since right-translation is an isometry, $\left\lfloor \frac{d_c}{e_k} \right\rfloor$ disjoint isometric copies of $T_k$ can be located on $\alpha$. By Lemma \ref{lm:covering}, each copy requires $\binom{N}{k}$ balls of diameter $b^{2k}$ to cover, implying:

\begin{thm}\label{thm:c_su}
	The criticality ansatz implies $C_{\operatorname{SU}(n)}(b^k) \geq \left\lfloor \frac{d_c}{e_k} \right\rfloor \binom{N}{k}$. \qed
\end{thm}

Let us compare this to the topological lower bound, call it $C^{\text{top}}(b^k) \coloneqq \frac{d_c}{ek} \binom{N}{k}$ with the trivial lower bound $C^t$ at $k$th powers of the base $b$.
\begin{equation}
	\frac{C^{\text{top}}(b^k)}{C^t(b^k)} = \frac{\frac{d_c}{d_k}\binom{N}{k}}{\frac{d_c}{b^k}} = \frac{b^k}{d_k} \binom{N}{k} = \frac{b^k}{\sqrt{\sum_{i=1}^k \binom{N}{i} b^{2k}}} \binom{N}{k}
\end{equation}

For $k < \sqrt{N}$ we may approximate this as:
\begin{equation}\label{eq:approx}
	\frac{C^{\text{top}}(b^k)}{C^t(b^k)} \approx \sqrt{\frac{b^{2k}\binom{N}{k}^2}{\binom{N}{k}b^{2k}}} = \sqrt{\binom{N}{k}}
\end{equation}

Line \ref{eq:approx} is the multiplicative enhancement the cohomological argument achieves over the naive geometric one.

The cohomological method has the potential for generalization beyond the abelian Lie algebras and torus subgroups employed. Any sub Lie group $G \subset \operatorname{SU}(n)$ which: (1) supports long nontrivial cup products, and (2) has its Lie algebra $\mathfrak{g}$ spanned by high weight directions can play the role assigned to the torus $T_k$ in our proof. Condition (1) ensures a large covering number for any acyclic covering of the subgroup and (2) ensures that even fairly large balls in the metric will indeed be acyclic. Looking for such Lie subgroups, actually their subalgebras, is a problem in linear coding theory. One may regard $\operatorname{su}(n)$ as a real Clifford algebra spanned by products of $2n$ gamma-matices (Majoranas) $\gamma_1, \dots, \gamma_{2n}$. Then a bit string of length $2n$ describes such a product and the condition that a set of such bit strings comprise a linear subspace $F_2^k \subset F_2^{2n}$, says that the strings constitute a \emph{linear code}, and algebraically that the operators these string span over $\R$ is both an associative algebra when $\cd =$ Clifford multiplication and a Lie algebra under $[x,y] = x \cd y - y \cd x$. In the case we studied, in qubit language, the bracket vanished on the Pauli words with just 1's and $X$'s, and we had Cartan subalgebras. The other extreme would be to locate Majorana bit strings $x,y$ which all anti-commute. This is the combinatorial condition $w(x) + w(y) - \operatorname{overlap}(x,y) = odd$. In that case, the Lie algebra spanned will be an $\operatorname{su}(k) \subset \operatorname{su}(n)$. In coding language $\frac{k}{n}$ is the \emph{rate} and $\min_{x \in F_2^k} w(x) = d$ the code's \emph{distance}.

Linear codes with these three properties: large rate, large distance, and $w(x) + w(y) - \operatorname{overlap}(x,y)$ odd would yield a nonabelian alternative to the estimates presented here. It is our hope such codes may be found.

We conclude with a remark on a path \emph{not} taken. If one has a Riemannian metric on a compact homogeneous space $M$ for which $\operatorname{vol}(r)$, the volume of a ball of radius $r$ about a (any) point of $M$, is fairly well understood, then there is a useful lower bound:
\begin{equation}\label{eq:vol_bound}
	C(d) > \frac{\operatorname{vol}(M)}{\operatorname{vol}(2r)}
\end{equation}

It is natural to ask if (\ref{eq:vol_bound}) could be applied to $\operatorname{SU}(n)_{\text{critical}}$. As we remarked at the start, $\operatorname{vol}(M)$ is accessible but known methods for upper bounding $\operatorname{vol}(2r)$ are very slack for $\operatorname{SU}(n)$ with exponential penalty metrics. The tool here is the classical Bishop-Gromov theorem; the problem with its application is it takes as input the \emph{smallest} eigenvalue of the Ricci quadratic form---not any sort of average value. Computations show that for exponential penalty metrics of $\operatorname{SU}(n)$, the smallest eigenvalue (in norm) grows exponentially with $N$. The spectrum has very long tails. This leads to a wild overestimate of volume growth as a function of $r$ making (\ref{eq:vol_bound}) inefficient. A better understanding of volume growth in right invariant metrics on $\operatorname{su}(n)$ is the subject of \cite{brown21} and \cites{BF22a,BF22b} but so far does not supercede the result presented here.

\bibliography{references}

\end{document}